\documentclass{article}
\usepackage{graphicx} 
\usepackage{amsmath}
\usepackage{amssymb}
\usepackage{amsthm}
\usepackage{amssymb}
\usepackage{mathdots}
\usepackage{enumitem}
\usepackage{tabularx}
\usepackage{color}
\title{Labeled histories for multifurcating trees}
\usepackage[margin=1in]{geometry}
\usepackage{stmaryrd}
\usepackage{lscape}

\renewcommand{\geq}{\geqslant}
\renewcommand{\leq}{\leqslant}

\newtheorem{theorem}{Theorem}
\newtheorem{prop}[theorem]{Proposition}

\newtheorem{lemma}[theorem]{Lemma}
\newtheorem{defi}[theorem]{Definition}
\newtheorem{conjecture}[theorem]{Conjecture}

\title{Labeled histories and maximally probable labeled topologies \\ with multifurcation}
\author{Emily H.~Dickey$^*$ \& Noah A.~Rosenberg\thanks{Department of Biology, Stanford University, Stanford, CA 94305 USA}}
\date{\today}
\begin{document}
\maketitle

\begin{abstract} \noindent In mathematical phylogenetics, labeled histories describe the sequences by which sets of labeled lineages coalesce to a shared ancestral lineage. We study labeled histories for at-most-$r$-furcating trees. Consider a rooted leaf-labeled tree in which internal nodes each have $i$ offspring, and $i$ is permitted to range from 2 to $r$ across internal nodes, for a specified value of $r$. For labeled topologies with $n$ leaves, we enumerate the total number of labeled histories with at-most-$r$-furcation. We enumerate the labeled histories possessed by a specific at-most-$r$-furcating labeled topology. We then demonstrate that the maximally probable at-most-$r$-furcating unlabeled topology on $n \geq 2$ leaves---the unlabeled topology whose labelings have the largest number of labeled histories---is the maximally probable strictly bifurcating unlabeled topology on $n$ leaves. Finally, we enumerate labeled histories for at-most-$r$-furcating labeled topologies in a setting that permits simultaneous branchings. We similarly reduce the problem of identifying the maximally probable at-most-$r$-furcating unlabeled topology on $n \geq 2$ leaves, allowing simultaneity, to that of identifying the maximally probable strictly bifurcating unlabeled topology on $n$ leaves, with simultaneity; we conjecture the shape of this bifurcating unlabeled topology. The computations contribute to the study of multifurcation, which arises in various biological processes, and they connect to analogous mathematical settings involving precedence-constrained scheduling.
\end{abstract}

\vspace{.1cm}

\vspace{.1cm}
\noindent{\bf Mathematics subject classification (MSC2020):} 05A15, 05C05, 92D15

\vspace{.1cm}
\noindent{\bf Author for correspondence:} Noah A.~Rosenberg. Email: noahr@stanford.edu.

\section{Introduction}


In evolutionary models that give rise to tree structures, the \emph{labeled history} is an important concept. For a leaf-labeled tree on $n$ leaves, a labeled history encodes the sequence of branchings that produce the specific labeled tree topology. The labeled histories produce a state space for probabilistic computations; the enumerations of labeled histories with $n$ leaves, and of labeled histories compatible with a specific $n$-leaf labeled topology, therefore assist in such computations. Under the Yule--Harding probability model on labeled histories \cite{Harding1971, Yule25}, each labeled history is equally likely to be produced by the process of evolutionary descent. 

Hammersley \& Grimmett~\cite{Hammersley74}, building on a conjecture of Harding~\cite{Harding1971}, described a sequence of bifurcating unlabeled tree shapes, growing with the number of leaves $n$, whose associated labeled topologies maximize the number of labeled histories among all labeled topologies at fixed numbers of leaves. Degnan \& Rosenberg~\cite{Degnan06} termed these unlabeled shapes, and their associated labeled topologies, \emph{maximally probable}: among labeled topologies with a fixed number of leaves, the labeled topologies with the maximally probable unlabeled shape have the highest probability under the Yule--Harding model. Degnan \& Rosenberg~\cite{Degnan06} used the maximally probable unlabeled tree shapes in a proof concerning gene tree and species tree labeled topologies. The maximally probable shapes have also appeared in various other phylogenetic combinatorics problems~\cite{Dickey25, DisantoAndRosenberg17}.

Much of the mathematical study of evolutionary trees has focused on bifurcating trees. However, multifurcating trees---in which internal nodes of a tree might possess more than two immediate offspring---arise in the context of models for phenomena such as epidemic transmission, large family sizes, high variance in reproductive success, and adaptive radiation~\cite{Eldon20, Maranca2023, Wakeley08}. Mathematically convenient formulations for multifurcating trees include strict \emph{$r$-furcation}, in which each internal node of a tree possesses precisely $r$ immediate descendants for a constant $r \geq 2$, and \emph{at-most-$r$-furcation}, in which internal nodes are permitted to vary from 2 to $r$ in their numbers of immediate descendants~\cite{Maranca2023}.

Previously, we extended classical enumeration results for labeled histories from bifurcating trees to strict $r$-furcation~\cite{Dickey25}. Specifically, (1) we enumerated the total number of labeled histories across all $r$-furcating labeled topologies with $n$ leaves, and (2) we enumerated labeled histories for a specific labeled topology with $n$ leaves. In addition, (3) we conjectured the maximally probable $r$-furcating unlabeled tree shape for $n$ leaves. We also considered problems (1) and (2) in the setting in which simultaneous branching events are permitted, counting ``tie-permitting labeled histories'' both with bifurcation and with strict $r$-furcation.

In this paper, we extend these results to at-most-$r$-furcation. For non-simultaneous at-most-$r$-furcation, we (1) we enumerate the total number of labeled histories across all at-most-$r$-furcating labeled topologies with $n$ leaves, and (2) we enumerate labeled histories for a specific at-most-$r$-furcating labeled topology with $n$ leaves. Next, (3) we show that the maximally probable at-most-$r$-furcating unlabeled tree shape on $n$ leaves is the maximally probable bifurcating unlabeled tree shape on $n$ leaves. For simultaneous at-most-$r$-furcation, we solve problems (1) and (2), enumerating the total number of tie-permitting labeled histories across all at-most-$r$-furcating labeled topologies and enumerating tie-permitting labeled histories for a specific at-most-$r$-furcating labeled topology on $n$ leaves. We (3) reduce the problem of identifying the maximally probable at-most-$r$-furcating unlabeled tree shape on $n$ leaves, with simultaneity, to that of identifying the maximally probable bifurcating unlabeled tree shape on $n$ leaves, with simultaneity. We present a conjecture describing this bifurcating shape. 

\section{Definitions}

Definitions largely follow Dickey \& Rosenberg \cite{Dickey25}, tracing to  Steel~\cite{Steel16} and King \& Rosenberg~\cite{King23}; definitions for at-most-$r$-furcating trees follow Maranca \& Rosenberg \cite{Maranca2023}. 

We consider leaf-labeled, rooted trees $T$. Each leaf has a distinct label. For a tree $T$, nodes are \emph{leaf nodes} or \emph{internal nodes}; the unique \emph{root node} is included among internal nodes. The \emph{labeled topology} of $T$ is its topological structure together with its leaf labels. The \emph{unlabeled tree shape} or \emph{unlabeled topology} of $T$ is the topological structure without the leaf labels. We indicate the number of leaves of $T$ by $|T|$.

For nodes $v$ and $u$ of $T$, $u$ is \emph{descended} from $v$ and $v$ is \emph{ancestral} to $u$ if $v$ lies on the path from the root to $u$. A node is ancestral to itself and descended from itself. 

A \emph{pendant edge} is an edge that connects an internal node to a leaf. A \emph{cherry node} is an internal node with exactly two child nodes, both of which are leaves. 

For $r \geq 2$, in an \emph{$r$-furcating tree}, also termed a \emph{strictly $r$-furcating tree}, each internal node has exactly $r$ immediate descendant nodes. The generalization of a cherry node for $r$-furcating trees is a \emph{broomstick node}, an internal node whose $r$ children are all leaves.

In an \emph{at-most-$r$-furcating tree}, the number of immediate descendant nodes of an internal node ranges from 2 to $r$ across internal nodes. Bifurcation corresponds to $r=2$. Strictly $r$-furcating trees are also at-most-$r$-furcating. For an at-most-$r$-furcating tree $T$ whose root has immediate subtrees $T_1, T_2, \ldots, T_k$, $2 \leq k \leq r$, we write $T = T_1 \oplus T_2 \oplus \ldots \oplus T_k$. 

We define $A_n$ as the set of at-most-$r$-furcating unlabeled tree shapes with $n$ leaves. The set of all at-most-$r$-furcating unlabeled tree shapes is 
\begin{align*}
    A = \bigcup_{n = 1}^\infty A_n.
\end{align*}
We have $A^* = A \backslash A_1$, where we disregard the trivial tree with one leaf. 

Let $m: A \cup \emptyset \to \mathbb{Z}^+$ be the function that extracts the number of leaves of a tree; for a tree $T$, $m(T)=|T|$. We define $m(\emptyset) = 0$. Let $s: A^* \to A \times A \times (A \cup \emptyset) \times \ldots \times (A \cup \emptyset)$, which maps an at-most-$r$-furcating tree to its immediate subtrees. The empty set is included as an option for subtrees $3, 4, \ldots, r$, as an at-most-$r$-furcating tree possesses at least two subtrees but need not possess more than two. We arrange subtrees of $s$ such that if the $i$th component of $s$, $s_i(T)$, satisfies $s_i(T) = \emptyset$, then $s_j(T) = \emptyset$ for all $j > i$. Finally, we define the function $w: A \cup \emptyset \to \mathbb{Z}^+$, which counts internal nodes of a tree, including the root. Again, we define $w(\emptyset) = 0$. For a strictly $r$-furcating tree $T$, we have
\begin{align*}
    w(T) = \frac{m(T) - 1}{r - 1}.
\end{align*}
The number of internal nodes is $w(T)=m(T)-1$ for bifurcating trees.
 
For a tree $T$, each node is associated with a \emph{time}. Leaves all have the same time. In the classic Yule--Harding model for bifurcating trees, internal nodes have distinct times, and the tree has \emph{non-simultaneous branching}. Given a labeled topology for a rooted tree $T$ with $w$ internal nodes and non-simultaneous branching, a \emph{labeled history} for $T$ is a bijection $f$ from the set of internal nodes of $T$ to $\{1, 2, \ldots, w(T)\}$, so that if node $u$ is descended from node $v$ in $T$ and $u \neq v$, then $f(u) < f(v)$. A labeled history can be viewed as the temporal sequence of internal nodes, with the convention here that the numbers assigned to nodes increase backward in time along genealogical lines, and the root is assigned the value $w(T)$.

The Yule--Harding model assumes that each internal node has a distinct time. If we modify the setting so that \emph{simultaneous branching} is allowed, internal nodes can possess the same time, and we modify the definition of a labeled history accordingly. An \emph{event} is a set of internal nodes that possess the same time. If $z$ is the number of events across all internal nodes, $z \leq w(T)$, then a labeled history $f$ for $T$ is a surjective function from the internal nodes of $T$ to $\{1, 2, \ldots, z\}$, such that if node $u$ is descended from node $v$ in $T$ and $u \neq v$, then $f(u) < f(v)$. The internal nodes of an event then map to the same element of $\{1, 2, \ldots, z\}$. Because simultaneity is permitted, the labeled history is not injective. If internal node $u$ is descended from $v$ and $u \neq v$, then $u$ and $v$ are not part of the same event. We sometimes use the term \emph{tie-permitting} to refer to labeled histories that allow simultaneity, as such labeled histories allow ``ties'' in node times.

Consider a fixed number of leaves $n$. For a given set of labeled topologies---bifurcating, strictly $r$-furcating, or at-most-$r$-furcating, without simultaneity or with simultaneity---a \emph{maximally probable} labeled topology is a labeled topology whose number of labeled histories is greater than or equal to that of all other labeled topologies~\cite{Degnan06}. Because each labeling of an unlabeled topology gives rise to the same number of labeled histories, we use \emph{unlabeled topologies} to indicate the maximally probable labeled topologies, and we refer to unlabeled topologies as maximally probable. The term \emph{maximally probable} arises from the fact that under the Yule--Harding probability model, the probability that a certain labeled topology is produced by the evolutionary process is proportional to its number of labeled histories---so that labeled topologies with the most labeled histories are the most probable labeled topologies under the model. Although the Yule--Harding model does not apply to sets of trees with multifurcation or simultaneity, we continue to use the term \emph{maximally probable} in these settings: a maximally probable unlabeled topology refers to an unlabeled topology whose labelings possess the largest number of labeled histories.  

Finally, we define a \emph{bifurcatable} tree as a tree $T$, with $n \geq 1$, for which each internal node has at most two non-leaf children. Trivially, every bifurcating tree is bifurcatable. The 1-leaf tree is trivially bifurcatable. Examples of bifurcatable trees appear in Figure \ref{fig:bifurcatable}. 

\begin{figure}
    \centering
\includegraphics[width=0.72\linewidth]{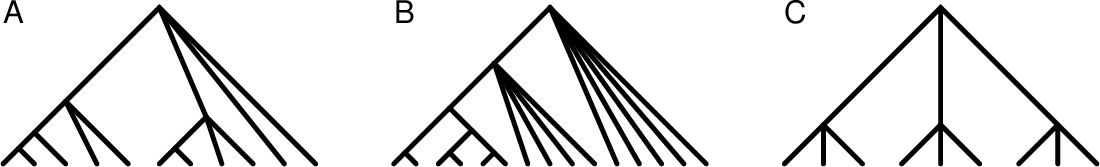}
    \caption{Bifurcatable and non-bifurcatable trees. (A) A bifurcatable at-most-trifurcating tree. (B) A bifurcatable at-most-6-furcating tree. (C) A non-bifurcatable at-most-trifurcating tree. The tree is non-bifurcatable because it possesses a node, the root node, that has more than two non-leaf child nodes.}
    \label{fig:bifurcatable}
\end{figure}

\section{Results}

For at-most-$r$-furcating trees $T$ with $n$ leaves, we examine (1) the number of labeled histories across all labeled topologies, (2) the number of labeled histories for a specific labeled topology, and (3) the characterization of maximally probable unlabeled tree shapes. Section \ref{sec:non-sim} proceeds under non-simultaneous branching. Section \ref{sec:sim} allows simultaneous branching.

\subsection{At-most-$r$-furcating trees, non-simultaneous branching}
\label{sec:non-sim}

Consider at-most-$r$-furcating trees with non-simultaneous bifurcations and $n\geq 2$ leaves. Fix $r$, $r \geq 2$. The trivial tree with a single leaf, $n=1$, is also permitted.

\subsubsection{Total number of labeled histories}

Let $A_r(n)$ denote the total number of labeled histories across all labeled topologies with $n$ leaves. Trivially, $A_r(1) = 1$. To count labeled histories, we proceed backward in time from the $n$ lineages. For $n\geq 2$, each group of $i$ lineages, $2 \leq i \leq \min\{n, r\}$, can be the first to coalesce, leaving $n-(i-1)$ lineages. The number of such groups that can be the first to coalesce is ${n \choose i}$. Proceeding recursively, the $n+1-i$ remaining lineages can coalesce in $A_r(n+1-i)$ ways. Summing over all possible values of $i$, we have the following result.
\begin{prop} \label{eq: num_at_most_r}
    Permitting only non-simultaneous at-most-$r$-furcations, the total number of labeled histories on $n$ leaves, $A_r(n)$, satisfies $A_r(1) = 1$, and for $n \geq 2$, 
    \begin{align*}
        A_r(n) = \sum_{i=2}^{\min\{n, r\}} {n \choose i} A_r(n+1-i).
    \end{align*}
\end{prop}
Note that it is convenient to allow $r>n$, although no additional labeled histories exist with $r>n$ relative to the case of $r=n$. With $r=2$, we have $A_2(n)={n \choose 2} A_2(n-1)$ for $n \geq 2$, and the recursion reduces to the recursion that counts labeled histories for strictly bifurcating trees \cite[Proposition 1]{Dickey25}. Applying the proposition recursively, for small values of $n$ and $r$, Table \ref{tab:r_nonsim} reports the values of $A_r(n)$.

\begin{table}[tb]
\centering
\small
\begin{tabular}{|c | c | c | c| c | c|}
\hline
       & \multicolumn{5}{c|}{$r$} \\ \cline{2-6} 
$n$    & 2 & 3 & 4 & 5 & 6 \\ \hline
1 &     1 &     1 &     1 &     1 &     1 \\ 
2  &             1 &             1 &             1 &             1 &             1 \\ 
3  &             3 &             4 &             4 &             4 &             4 \\ 
4  &            18 &            28 &            29 &            29 &            29 \\ 
5  &           180 &           320 &           335 &           336 &           336 \\ 
6  &         2,700 &         5,360 &         5,665 &         5,686 &         5,687 \\ 
7  &        56,700 &       123,760 &       131,705 &       132,265 &       132,293 \\ 
8  &     1,587,600 &     3,765,440 &     4,028,430 &     4,046,980 &     4,047,932 \\ 
9  &    57,153,600 &   145,951,680 &   156,800,490 &   157,560,312 &   157,599,498 \\ 
10 & 2,571,912,000 & 7,019,678,400 & 7,567,091,700 & 7,605,060,162 & 7,607,014,464 \\ \hline
\end{tabular}
\caption{The total number of labeled histories for at-most-$r$-furcating trees with $n$ leaves, $A_r(n)$, as obtained by Proposition \ref{eq: num_at_most_r}. The $r=2$ column, $A_2(n)$, corresponds to OEIS sequence A006472, and the $r=3$ column, $A_3(n)$, is OEIS sequence A358072. The diagonal $A_n(n)$ is A256006. For $r \geq n$, $A_r(n) = A_n(n)$. The table accords with Table 1 of \cite{Wirtz24}, which reported the first terms of $A_2(n)$, $A_3(n)$, $A_4(n)$, and $A_n(n)$.}
\label{tab:r_nonsim}
\end{table}

\subsubsection{Number of labeled histories for a specific topology}

Next, we enumerate labeled histories for a specific at-most-$r$-furcating labeled topology $T$ with non-simultaneous branchings and $n$ leaves. Let $s(T) = (T_1, T_2, \ldots, T_r)$ be the $r$ immediate subtrees of the root of $T$, with $\big(w(T_1), w(T_2), \ldots, w(T_r)\big)$ internal nodes, respectively. Note that $w(T) = w(T_1) + w(T_2) + \ldots + w(T_r) + 1$.

The computation is analogous to the case of strict $r$-furcation~\cite[eq.~3.6]{Dickey25}. For $n\geq 2$, the number of labeled histories, $N(T)$, is obtained recursively by
\begin{align} \label{eq: LH_at_most_r}
    N(T) = {w(T) - 1 \choose w(T_1), w(T_2), \ldots, w(T_r)} \, N(T_1) \, N(T_2) \,  \cdots N(T_r),
\end{align}
with $N(T) = 1$ for $|T| = 1$ and $T = \emptyset$. 

To obtain a non-recursive formula, define $V^0(T)$ as the set of internal nodes of $T$, including the root. We expand \eqref{eq: LH_at_most_r} and multiply by $w(T)/w(T)$.
\begin{prop}[\cite{Semple03}, p.~23]
    \label{prop:non-sim_specific}
    Permitting only non-simultaneous at-most-$r$-furcations, the number of labeled histories for a labeled topology $T$ with $n$ leaves satisfies $N(T) = 1$ for $n=1$ or $T = \emptyset$, and for $n \geq 2$,
    \begin{align*}
    N(T) = \frac{\big(w(T) \big)!}{\prod_{v \in V^0(T)} w(v)},
\end{align*}
where $w(v)$ is the number of internal nodes in the subtree of $T$ rooted at $v$, including $v$ itself.
\end{prop}
With $m(v)$ equal to the number of leaves descended from $v$, if $w(v) = \big(m(v)-1\big)/(r-1)$ for all $v \in V^0(T)$, then $T$ is strictly $r$-furcating. For the strictly $r$-furcating case, Proposition \ref{prop:non-sim_specific} recovers Proposition 8 of \cite{Dickey25}:
\begin{equation}
\label{eq:semple}
N(T) = \frac{ \big( \frac{n-1}{r-1} \big)! }{\prod_{v \in V^0(T)} \big( \frac{m(v)-1}{r-1} \big) }.
\end{equation}

\subsubsection{Maximally probable at-most-$r$-furcating labeled topologies} 
\label{sec: max_prob_at_most}

In this section, we show that the unique maximally probable at-most-$r$-furcating tree shape with $n$ leaves is the unique maximally probable strictly bifurcating tree shape. We first introduce a transformation, \emph{bifurcatization}, that converts an arbitrary at-most-$r$-furcating tree with $n$ leaves into a bifurcatable at-most-$r$-furcating tree with $n$ leaves. We then show that bifurcatization cannot decrease the number of labeled histories, from which we conclude that the strictly bifurcating tree is maximally probable.

First, we introduce a transformation for trees that we call \emph{pendant-pruning}.
\begin{defi} \label{defi:pendant-pruning}
Consider an at-most-$r$-furcating tree $T$ and a node $v$. The \emph{pendant-pruning} transformation $\mathcal{P}_v(T)$ applied to $T$ at node $v$ is defined as follows: 
\begin{enumerate}[label=(\roman*)]
    \item If node $v$ possesses 2 or fewer children, then $\mathcal{P}_v(T)=T$.
    \item If node $v$ possesses 3 or more children, at least 2 of which are internal nodes, then $\mathcal{P}_v(T)$ is obtained from $T$ by pruning all pendant edges descended from $v$ and their associated leaves.
    \item If node $v$ possesses 3 or more children, exactly 0 or 1 of which is an internal node, then $\mathcal{P}_v(T)$ is obtained from $T$ by pruning pendant edges descended from $v$ and their associated leaves until exactly 2 children of $v$ remain. 
\end{enumerate}
\end{defi}
Because only leaves are removed, pendant-pruning does not change the number of descendant internal nodes for an internal node $v$. If node $v$ possesses $r=2$ immediate descendants (condition (i)), or if $v$ possesses $r > 2$ descendants and the number of pendant edges is less than or equal to $r-2$ (condition (ii)), then pendant-pruning has a unique result. If $v$ possesses $r > 2$ descendants and the number of pendant edges is $r$ or $r-1$ (condition (iii)), then the pendant-pruned tree is not uniquely specified. 

Recall that an at-most-$r$-furcating tree is bifurcatable if each internal node has at most two non-leaf children. Consider a bifurcatable tree, $T$. Apply pendant-pruning to all of its internal nodes. A strictly bifurcating tree is produced, $T'$. Because pendant-pruning does not alter the number of descendant internal nodes for an internal node $v$, by Proposition \ref{prop:non-sim_specific}, $N(T') = N(T)$. We have obtained the following result.

\begin{lemma} \label{thm: bi_equal}
    For every \emph{bifurcatable} tree $T$ with $n \geq 1$ leaves, there exists an associated pendant-pruned bifurcating tree $T'$ that is obtained by pendant-pruning at each of the internal nodes of $T$ and that has the same number of labeled histories, $N(T')=N(T)$.
\end{lemma}
We now introduce a transformation that we call \emph{bifurcatization}, which converts an at-most-$r$-furcating tree into a bifurcatable tree. If a tree $T$ is bifurcatable, then bifurcatization of $T$ results simply in $T$. Otherwise, consider an internal node $v$ of $T$ that has at least 3 non-leaf subtrees: that is, if $T(v)$ is the subtree of $T$ rooted at $v$, then $s\big(T(v)\big) =(T_1, T_2, \ldots, T_k)$, where $T_1, T_2, T_3 \not = \emptyset$. Choose two of these subtrees, say $T_1$ and $T_2$, prune them from $v$, and construct a new subtree $T_1' = T_1 \oplus T_2$ descended from $v$.  We now have $s(v) = (T_1', T_3, \ldots, T_k, \emptyset)$, and the number of non-leaf children of node $v$ has decreased by 1. Formally, we have the following definition.
\begin{defi} \label{defi:bifurcatization} Consider an at-most-$r$-furcating tree $T$ and a node $v$. A \emph{bifurcatization} $\mathcal{B}_v(T)$ of $T$ at node $v$ is defined as follows: 
\begin{enumerate}[label=(\roman*)]
    \item If node $v$ possesses 2 non-leaf subtrees, then $\mathcal{B}_v(T)=T$.
    \item If node $v$ possesses at least 3 non-leaf subtrees $T_1, T_2, \ldots, T_k$, then $\mathcal{B}_v(T)$ is obtained from $T$ by replacing two of these subtrees, say $T_1$ and $T_2$, by a new subtree $T_1'$ from whose root $T_1$ and $T_2$ are appended. 
\end{enumerate}
\end{defi}
A (non-trivial) bifurcatization adds an internal node; if node $v$ possesses at least 3 non-leaf subtrees, then the number of internal nodes of $\mathcal{B}_v(T)$ exceeds that of $T$ by 1. Note that if a node $v$ possesses at least 3 non-leaf subtrees, then multiple bifurcatizations $\mathcal{B}_v(T)$ exist. Figure \ref{fig:transformation} shows bifurcatization applied to a tree twice sequentially, producing a bifurcatable tree.
\begin{figure}
    \centering
    \includegraphics[width=0.75\linewidth]{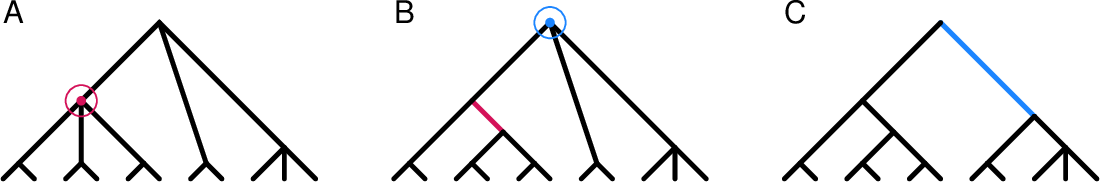}
    \caption{Bifurcatization. (A) A non-bifurcatable tree. (B) A tree produced by bifurcatization at the red internal node of (A). (C) A tree produced by bifurcatization at the blue internal node of (B). The tree in (C) is bifurcatable.  Bifurcatization is used in proving results \ref{thm:transform_more_LH}, \ref{thm:bi_max_prob}, \ref{thm:sim_transformation}, and \ref{thm:sim_max_bi}.}
    \label{fig:transformation}
\end{figure}

We now show that bifurcatization applied to a non-bifurcatable at-most-$r$-furcating tree $T$ increases the number of labeled histories.

\begin{lemma} \label{thm:transform_more_LH}
Consider a non-bifurcatable at-most-$r$-furcating tree $T$, and suppose $T'$ is obtained from $T$ by bifurcatization. Then $N(T') > N(T)$. 
\end{lemma}
\begin{proof}
     As $T$ is non-bifurcatable, consider an internal node $v$ that has at least 3 non-leaf subtrees, and call their roots $c_1, c_2, \ldots, c_m$, where $j \geq 3$. Let $T'= \mathcal{B}_v(T)$ be a bifurcatization of $T$ at node $v$. Let $k$ be the new internal node created in the bifurcatization, with $c_1$ and $c_2$ as the children of $k$. The number of internal nodes $w(T')$ satisfies $w(T')=w(T)+1$. 

    First, we show that for each labeled history of $T$, we can construct a labeled history of $T'$. Consider a labeled history $f$ of $T$, with labels $\{1, 2, \ldots, w(T)\}$ as the image of the internal nodes of $T$. We construct a labeled history $f'$ for $T'$, with labels $\{1, 2, \ldots, w(T), w(T)+1\}$ as the image of the internal nodes of $T'$; an example appears in Figure \ref{fig:Transformation_LH_matching}. For the node $v$, at which the bifurcatization takes place, suppose the numerical label is $f(v) = i$. For $i \leq \ell \leq w(T)$, we set $f'^{-1}(\ell + 1) = f^{-1}(\ell)$. That is, we increment the label by 1 for all nodes with labels larger than or equal to that of $v$, so that the relative ordering of these nodes remains unchanged. We let $f'(k) = f(v) = i$. We retain $f'^{-1}(\ell) = f^{-1}(\ell)$ for $1 \leq \ell \leq i-1$; again, the relative ordering of these nodes remains unchanged. Because for each labeled history of $T$, we can construct at least one corresponding labeled history for $T'$, $N(T) \leq N(T')$.

    Next, to prove that the inequality is strict, we show that there exists some labeled history of $T$ that can be associated with \emph{multiple} labeled histories for $T'$. Because $c_1, c_2, \ldots, c_m$ share the same parent node in $T$, there must exist a labeled history of $T$ in which $f(c_3) > f(c_2)$ and $f(c_3) > f(c_1)$. Let $f(c_3) = j$, and for $j+1 \leq \ell \leq \omega(T)$, we set $f'^{-1}(\ell + 1) = f^{-1}(\ell)$. That is, we increment the label by 1 for all nodes with labels larger than that of $c_3$, so that the relative ordering of these nodes remains unchanged. For our first labeled history, we set $f'^{-1}(c_3) = j + 1$, $f^{-1}(k) = j$, and retain $f'^{-1}(\ell) = f^{-1}(\ell)$ for $1 \leq \ell \leq j-1$. For our second labeled history, we set $f'^{-1}(k) = j + 1$, $f^{-1}(c_3) = j$, and retain $f'^{-1}(\ell) = f^{-1}(\ell)$ for $1 \leq \ell \leq j-1$.
        
    Because for each labeled history of $T$, we can construct at least one corresponding labeled history for $T'$ and because there exists at least one labeled history for which we can construct more than one corresponding labeled history for $T'$, $N(T) < N(T')$.
\end{proof}

\begin{figure}
    \centering
    \includegraphics[width=0.6\linewidth]{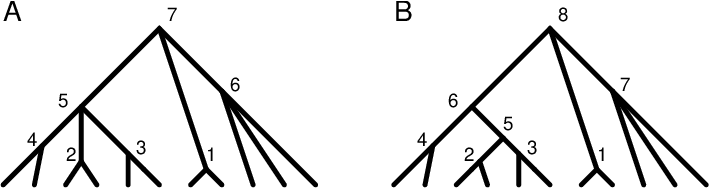}
    \caption{The bifurcatization in the proof of Lemma \ref{thm:transform_more_LH}. (A) Tree $T$. (B) Tree $T'$ after the bifurcatization. The lemma constructs a labeled history for $T'$ from a labeled history for $T$; in (B), the node 5 corresponds to node $k$ in the proof of the lemma.}
    \label{fig:Transformation_LH_matching}
\end{figure}

With Lemmas \ref{thm: bi_equal} and \ref{thm:transform_more_LH}, we can show that the unique maximally probable at-most-$r$-furcating tree with $n$ leaves is the unique strictly bifurcating tree with $n$ leaves. 

We first recall the form of the unique maximally probable strictly bifurcating tree with $n \geq 3$ leaves. 
\begin{theorem}[\cite{Hammersley74}]
\label{thm: hammersley}
Permitting only non-simultaneous bifurcations, the unique unlabeled topology whose labelings have the largest number of labeled histories among unlabeled topologies with $n$ leaves takes the form $U_n^* = U_t^* \oplus U_{n-t}^*$, where for $n \geq 3$,
\begin{align*}
t = 2^{ \lfloor \log_2 (\frac{n-1}{3})\rfloor + 1}.
\end{align*}
\end{theorem}
The topology is obtained by a decomposition at the root into trees of size $t$ equal to a certain power of 2, and the remainder equal to $n-t$. We have $(t,n-t)=(1,1)$ for the trivial $n=2$, and for $n=3$ to $n=16$, $(t,n-t)=(1,2)$, $(2,2)$, $(2,3)$, $(2,4)$, $(4,3)$, $(4,4)$, $(4,5)$, $(4,6)$, $(4,7)$, $(4,8)$, $(8,5)$, $(8,6)$, $(8,7)$, and $(8,8)$. 

Note that by eq.~\ref{eq: LH_at_most_r}, for maximally probable unlabeled topologies with $n \geq 3$ leaves, we can show that the number of labeled histories is strictly monotonically increasing with the number of leaves $n$. 

\begin{prop}
\label{prop:monotonic}
For $n \geq 3$ leaves, the number of labeled histories of the unique maximally probable bifurcating unlabeled topology, $N(U_n^*)$, increases monotonically with $n$.
\end{prop}
\begin{proof}
We proceed by induction. It is convenient to establish monotonicity for the first several base cases by direct calculation: $N(U_3^*) = 1$, $N(U_4^*) = 2$, $N(U_5^*) = 3$, $N(U_6^*) = 8$, $N(U_7^*) = 20$, $N(U_8^*) = 80$.

Suppose that for all $i$ such that $3 \leq i \leq k$, $N(U_i^*) < N(U_{i+1}^*)$. We prove $N(U_{k+1}^*) < N(U_{k+2}^*)$. Note that based on the sequence of base cases, we can assume $k+1 \geq 8$, so that $t_k = 2^{\lfloor \log_2(\frac{k}{3}) \rfloor + 1} \geq 4$.

By Theorem \ref{thm: hammersley}, $U_{k+1}^* = U_{t_{k+1}}^* \oplus U_{k+1 - t_{k+1}}^*$, where $t_{k+1} = 2^{\lfloor \log_2(\frac{k}{3}) \rfloor + 1}$, and $U_{k+2}^* = U_{t_{k+2}}^* \oplus U_{k+2 - t_{k+2}}^*$, where $t_{k+2} = 2^{\lfloor \log_2(\frac{k+1}{3}) \rfloor + 1} $. Note that $t_{k+1}, t_{k+2} \geq 4$ because $k+1 \geq 8$. 

We first demonstrate that a subtree size is shared by $U_{k+1}^*$ and $U_{k+2}^*$. There are two cases. (i) If $t_{k+1} = t_{k+2}$, then  $|U_{t_{k+1}}^*| = |U_{t_{k+2}}^*|$. Otherwise, (ii) if $t_{k+1} \neq t_{k+2}$, then incrementing the tree size by 1 increments the exponent of the size of the ``left'' subtree by 1, so that $t_{k+2} = 2t_{k+1}$. This case requires $\lfloor \log_2(\frac{k+1}{3}) \rfloor = \lfloor \log_2(\frac{k}{3}) \rfloor + 1$. Writing $\ell = \lfloor \log_2(\frac{k+1}{3}) \rfloor$, $\log_2(\frac{k+1}{3})$ must be an integer, or $\frac{k+1}{3} = 2^{\ell}$, so that for case (ii) to apply, $k+2 = 3 \times 2^\ell + 1$. Because $t_{k+2} = 2^{\ell + 1}$, the other subtree of $U_{k+2}^*$ has size $k+2-t_{k+2} = 3 \times 2^\ell + 1 - 2^{\ell + 1} = 2^{\ell} + 1$ leaves. The subtrees of $U_{k+1}^*$ have sizes $2^\ell$ and $k+1-2^\ell = 3 \times 2^\ell -2^\ell = 2^{\ell+1}$, so that both $U_{k+2}^*$ and $U_{k+1}^*$ have a subtree of size $2^{\ell+1}$.

Next, by eq.~\ref{eq: LH_at_most_r}, 
\begin{align}
\label{eq:Nk2}
N(U_{k+2}^*) &= {k \choose t_{k+2}-1} \, N(U_{t_{k+2}}^*) \, N(U_{k+2-t_{k+2}}^*), \\
\label{eq:Nk1}
N(U_{k+1}^*) &= {k-1 \choose t_{k+1}-1} \, N(U_{t_{k+1}}^*) \, N(U_{k+1-t_{k+1}}^*).
\end{align}
For convenience, write the binomial coefficients $B_{k+2}={k \choose t_{k+2}-1}$ and $B_{k+1}= {k-1 \choose t_{k+1}-1}$. 

Continuing the two cases above, if (i) $t_{k+1} = t_{k+2}$, then $B_{k+2} > B_{k+1}$ because ${n \choose r} > {n-1 \choose r}$ for $n > r$. We have $N(U_{t_{k+2}}^*) = N(U_{t_{k+1}}^*)$, and $N(U_{k+2-t_{k+2}}^*) > N(U_{k+1-t_{k+1}}^*)$ by the inductive hypothesis. Using eqs.~\ref{eq:Nk2} and \ref{eq:Nk1}, we conclude $N(U_{k+2}^*) = B_{k+2} \, N(U_{t_{k+2}}^*) \, N(U_{k+2-t_{k+2}}^*)  >  B_{k+1} \, N(U_{t_{k+2}}^*) \, N(U_{k+1-t_{k+1}}^*) = N(U_{k+1}^*)$.

Otherwise, if (ii) $t_{k+2}=2t_{k+1}$, then $k+2 = 3 \times 2^\ell + 1$, $t_{k+2}=k+1-2^\ell = 2^{\ell+1}$, and 
\begin{align*}
N(U_{k+2}^*) &= \frac{(3 \times 2^\ell - 1)!}{(2^{\ell+1}-1)! \, (2^{\ell})!}
N(U_{2^{\ell+1}}^*) \, N(U_{2^{\ell}+1}^*), \\
N(U_{k+1}^*) &= \frac{(3 \times 2^\ell - 2)!}{(2^\ell - 1)! \, (2^{\ell+1}-1)!}
N(U_{2^\ell}^*) \, N(U_{2^{\ell+1}}^*). 
\end{align*}
Taking the ratio of $N(U_{k+2}^*)$ and $N(U_{k+1}^*)$, we apply the inductive hypothesis to trees of size $2^\ell$, noting $2^\ell \geq 4$ for $k+1 \geq 8$, and we obtain 
$$\frac{N(U_{k+2}^*)}{N(U_{k+1}^*)} = \bigg( \frac{3 \times 2^\ell - 1}{2^\ell} \bigg) \bigg( \frac{ N(U_{2^{\ell}+1}^*)}{N(U_{2^{\ell}}^*)} \bigg) > 1.$$
The induction is now complete.
\end{proof}

\begin{theorem}
\label{thm:bi_max_prob}
    For $r \geq 2$ and $n \geq 1$, the maximally probable at-most-$r$-furcating tree topology with $n$ leaves is unique, and it is the unique maximally probable strictly bifurcating tree with $n$ leaves.
\end{theorem}
\begin{proof}
For $n=1$ and $n=2$, the claim is trivial, as there is only one possible tree shape. For $n \geq 3$, let $\hat{T}_n=U_n^*$ be the unique maximally probable bifurcating tree with $n$ leaves, as defined by Theorem \ref{thm: hammersley}. Consider an at-most-$r$-furcating tree, $T$, with $n$ leaves. There are three cases: (i) $T$ is a strictly bifurcating tree, (ii) $T$ is bifurcatable but not strictly bifurcating, and (iii) $T$ is non-bifurcatable. 

(i) If $T$ is strictly bifurcating, then it is either the maximally probable strictly bifurcating tree, in which case $N(T) = N(\hat{T}_n)$, or it is not the maximally probable strictly bifurcating tree, in which case $N(T) < N(\hat{T}_n)$ by the uniqueness of the maximally probable strictly bifurcating tree.

(ii) If $T$ is bifurcatable and not strictly bifurcating, then we apply pendant-pruning to produce $T^*$, a bifurcating tree with $n^* < n$ leaves. By Lemma \ref{thm: bi_equal}, we have $N(T) = N(T^*)$, and by the monotonicity in Proposition \ref{prop:monotonic}, $N(T^*) \leq N(\hat{T}_{n^*}) < N(\hat{T}_n)$, so that $N(T) < N(\hat{T}_n)$.

(iii) If $T$ is non-bifurcatable, then we sequentially apply bifurcatization until a bifurcatable tree $T'$ is reached. By Lemma \ref{thm:transform_more_LH}, $N(T) < N(T')$. $T'$ is strictly bifurcating or it is bifurcatable but not strictly bifurcating, so that case (i) or (ii) applies to $T'$ and $N(T') \leq N(\hat{T}_n)$. Because $N(T) < N(T')$, we conclude $N(T) < N(\hat{T}_n)$. 
     
We conclude that $N(T) \leq N(\hat{T}_n)$, with equality if and only if $T=\hat{T}_n$.
\end{proof}

\subsection{At-most-$r$-furcating trees, simultaneous branching}
\label{sec:sim}

\subsubsection{Total number of labeled histories}

Let $S_r(n)$ denote the total number of tie-permitting labeled histories across all labeled topologies with $n$ leaves. Trivially, $S_r(1) = 1$. We count labeled histories by proceeding backward in time from the $n$ lineages. We choose groups of lineages to coalesce simultaneously in the first ``event.'' Each such group has a number of lineages in $\{2, 3, \ldots, r\}$ . Specifically, let $x_i$ be the number of groups of size $i$, $2 \leq i \leq r$, that coalesce in this event. We must have $2 \leq 2x_2 + 3x_3 + \ldots +r x_r \leq n$; that is, at least 2 and at most $n$ lineages coalesce. With fixed $x_2, x_3, \ldots, x_r$, without loss of generality, the groups are selected in ascending size. For groups of size $i$, there are ${n - \sum_{j=2}^{i-1} j x_j \choose i}$ choices for the first group, ${n - (\sum_{j=2}^{i-1} j x_j ) - i \choose i}$ for the second group, and so on, with ${n - (\sum_{j=2}^{i-1} j x_j ) - i(x_i - 1) \choose i}$ for the $x_i$th group. The same $x_i$ groups can be chosen in any of $x_i!$ orders. We obtain
\begin{align*} 
& \bigg[ \frac{1}{x_2!} {n \choose 2} {n-2 \choose 2} \times \cdots \times {n-2(x_2-1) \choose 2} \bigg] \nonumber \\
& \quad \times \bigg[ \frac{1}{x_3!} {n - 2x_2 \choose 3} {n-2x_2 - 3 \choose 3} \times \cdots \times {n-2x_2-3(x_3-1) \choose 3} \bigg] \nonumber \\
& \quad \times \cdots \times \bigg[ \frac{1}{x_r!} {n - 2x_2 - 3x_3 - \ldots - (r-1)x_{r-1} \choose r} {n - 2x_2 - 3x_3 - \ldots - (r-1)x_{r-1} - r \choose r}  \nonumber \\
& \quad \times \cdots \times {n - 2x_2 - 3x_3 - \ldots - (r-1)x_{r-1} - r(x_r-1) \choose r} \bigg] \nonumber \\
& = 
\frac{1}{x_2! \, x_3! \cdots x_r!} \frac{n!}{(2!)^{x_2} (3!)^{x_3} \cdots (r!)^{x_r} \big(n- 2x_2 - 3x_3 - \ldots - (r-1)x_{r-1} - rx_r \big)!}.
\end{align*}
Once all groups of all sizes coalesce in this event, we are left with $n- \sum_{j=2}^r (j-1) x_j$ lineages. The number of ways in which these remaining lineages can coalesce is $S_r \big(n- \sum_{j=2}^r (j-1)x_j \big)$.
\begin{prop} \label{eq:arb_at_most_r_sim}
    Permitting simultaneous at-most-$r$-furcations, the total number of labeled histories on $n$ leaves, $S_r(n)$, satisfies $S_r(1) =1$, and for $n \geq 2$,
    \begin{align*}
        S_r(n) & = {\sum}_{\{ (x_2, x_3, \ldots, x_r) \, : \,  2 \leq \sum_{j=2}^r jx_j \leq n \} }   
        \frac{n!}{\big[ \prod_{j=2}^r (j!)^{x_j} (x_j!) \big]
(n- \sum_{j=2}^r jx_j)!} S_r\bigg(n - \sum_{j=2}^r (j-1)x_j\bigg).        
    \end{align*}
\end{prop}

\begin{table}[tb]
\centering
\small
\begin{tabular}{|c | c | c | c| c | c|}
\hline
      & \multicolumn{5}{c|}{$r$} \\ \cline{2-6}
$n$ & 2 & 3 & 4 & 5 & 6 \\ \hline
1   &              1 &              1 &              1 &              1 &              1 \\ 
2   &              1 &              1 &              1 &              1 &              1 \\ 
3   &              3 &              4 &              4 &              4 &              4 \\ 
4   &             21 &             31 &             32 &             32 &             32 \\ 
5   &            255 &            420 &            435 &            436 &            436 \\ 
6   &          4,815 &          8,625 &          8,990 &          9,011 &          9,012 \\
7   &        130,095 &        250,390 &        262,045 &        262,731 &        262,759 \\ 
8   &      4,763,115 &      9,755,865 &     10,241,105 &     10,269,469 &     10,270,659 \\ 
9   &    226,955,925 &    491,081,920 &    516,730,165 &    518,213,374 &    518,275,576 \\ 
10  & 13,646,570,175 & 31,014,100,075 & 32,695,631,150 & 32,791,782,443 & 32,795,799,470 \\ \hline
\end{tabular}
\caption{The total number of labeled histories, allowing simultaneity, for an at-most-$r$-furcating tree with $n$ leaves, $S_r(n)$, as obtained by Proposition \ref{eq:arb_at_most_r_sim}. The $r=2$ column, $S_2(n)$, corresponds to OEIS sequence A317059. The diagonal $S_n(n)$ is A005121. As in Table \ref{tab:r_nonsim}, for $r\geq n$, $S_r(n) = S_n(n)$.}
\label{tab:r_sim}
\end{table}

In the proposition, the sum traverses the partitions of $n$ that only contain entries of $2, 3, \ldots, r$ (OEIS A002865 in the calculation of $S_n(n)$). Table \ref{tab:r_sim} gives the values of $S_r(n)$ for small values of $n$ and $r$. We observe that $S_2(n)$ reduces to the number of labeled histories for bifurcating trees with simultaneity, Proposition 5 from \cite{Dickey25}: $x_2$ is the only nonzero element in $(x_2, x_3, \ldots, x_r$), so that the sum becomes 
$$S_2(n)=\sum_{x_2=1}^{\lfloor n/2 \rfloor } \frac{n!}{2^{x_2} \, x_2! \, (n-2x_2)!} S_2(n-x_2).$$

\subsubsection{Number of labeled histories for a specific topology}

Next, we count the labeled histories for a labeled topology $T$ of an at-most-$r$-furcating tree allowing simultaneity, extending Theorem 15 from \cite{Dickey25}. The proof follows closely.

Write $E(T, z)$ for the number of tie-permitting labeled histories of labeled topology $T$ with $z$ events. For labeled topology $T$, the maximal number of events is the number of internal nodes $w(T)$, if each internal node occurs at a distinct time point. The minimal number of permissible events is $\delta(T)$, the height of tree $T$. Hence, for labeled topology, $T$, the number of events, $z$, must satisfy
$$\delta(T) \leq z \leq w(T).$$

Consider an at-most-$r$-furcating tree $T$. If $T$ has $r$ non-empty subtrees of the root, then there exist $2^r - 1$ possible (non-empty) sets of subtrees. Each of these sets provides a possible  ``event type'': a set of subtrees that can be associated with an event as the subtrees in which simultaneous nodes associated with that event occur. We encode event type $k$, $1 \leq k \leq 2^r - 1$, in binary, with $r$ digits. Reading left to right, the $j$th digit of the binary representation of $k$, $1 \leq j \leq r$, indicates if an event occurs in subtree $j$. An event of type $k$ possesses simultaneous at-most-$r$-furcations in all subtrees for which the binary representation of $k$ has a 1. 

The $z-1$ non-root events each have a type among $1, 2, \ldots, 2^r -1$. Each labeled history has a ``simultaneity configuration,'' a vector that counts the numbers of events of the different types. We write the simultaneity configuration $c^*=(c_1-1, c_2-1, \ldots, c_{2^r-1}-1)$, where $c_k-1$ counts events of type $k$. It is convenient to use $c_k-1$ rather than $c_k$ to count events of type $k$, as the vector $c=(c_1,c_2,\ldots, c_{2^r-1})$ is then a composition of $(z-1)+(2^r-1)$ into $2^r-1$ ordered, positive integer parts.

Write $I(c,j) = \sum_{k=1}^{2^r-1} c_k^* f(k,j)$, where $f(k,j)=1$ if the $r$-digit binary representation of integer $k$ has a 1 in position $j$. $I(c,j)$ counts internal nodes of subtree $j$ for a simultaneity configuration $c^*$ encoded by composition $c$. For a simultaneity configuration $c^*$ that has specified numbers of events $a_1, a_2, \ldots, a_r$ in subtrees $1, 2, \ldots, r$, the number of labeled histories for subtree $j$ is $E(T_j, a_j)$.

If the at-most-$r$-furcating tree $T$ possesses only $b$ immediate subtrees of the root, $2 \leq b < r$, then we view the tree as having $r-b$ empty subtrees; the permissible event types $k$ are only those integers $1 \leq k < 2^r-1$ whose binary representations have a 0 in the last $r-b$ positions, so that for each composition $c$ and associated simultaneity configuration $c^*$, $I(c,j)=0$ for $b+1 \leq j \leq r$. The number of labeled histories can be written using the same recursion as for the strictly $r$-furcating case, with $w(T)$ providing a general expression for the count of internal nodes. Sums over $a_j$ with $b+1 \leq j \leq r$ collapse to $a_j=0$ and the sum over $c$ requires consideration only of compositions with $(c_{b+1}^*, c_{b+2}^*, \ldots, c_r^*)=(0,0,\ldots,0)$.

\begin{prop} 
\label{eq:sim_for_at_most_r}
Permitting simultaneous at-most-$r$-furcations, the number of labeled histories for a labeled topology $T$ with $n$ leaves, $N(T)$,
satisfies
$$N(T) = \sum_{z = \delta(T)}^{w(T)} E(T,z).$$
The number of tie-permitting labeled histories $E(T,z)$ satisfies
\begin{enumerate}[label=(\roman*)]
    \item If $T$ has 1 leaf or is empty, then $E(T,0)=1$ and $E(T,z)=0$ for $z \neq 0$.
    \item If $|T_1| \geq 1$ and $|T_2| \geq 1$ and $\max_{1 \leq j \leq r} |T_j| = 1$, then $E(T,1) = 1$ and $E(T,z) = 0$ for $z \not = 1$.
    \item If $\max_{1 \leq j \leq r} |T_j| > 1$, then 

\begin{align*}
    E (T, z) &= \sum_{a_1 = \delta(T_1)}^{\min \big(w(T_1), z-1 \big)} \sum_{a_2 = \delta(T_2)}^{\min \big(w(T_2), z-1 \big)} \ldots \sum_{a_r = \delta(T_r)}^{\min \big(w(T_r), z-1 \big)} \sum_{c \in C(z+2^r - 2, 2^r - 1)} \prod_{j=1}^r \llbracket I(c, j) = a_j\rrbracket \nonumber \\
& \qquad\times \bigg( \prod_{j=1}^r E(T_j,a_j) \bigg) {z-1 \choose c_1^*, c_2^*, \ldots, c_{2^r-1}^* }.
\end{align*}
\end{enumerate}
\end{prop}
\noindent The Iverson bracket $\llbracket \cdot \rrbracket$ equals 1 if its statement holds and is 0 otherwise.  

\subsubsection{Maximally probable at-most-$r$-furcating tree with simultaneity}
\label{sec:max-at-most-r-simul}

Here, we reduce the problem of identifying the maximally probable at-most-$r$-furcating tree shape with simultaneity to that of determining the maximally probable bifurcating tree shape with simultaneity. With simultaneity, a maximally probable tree is a topology, $T$, that maximizes $N(T)$. A topology can also maximize $E(T, z)$, the number of labeled histories with a fixed number of leaves $n$ and events $z$. We follow similar logic to our corresponding result in Section \ref{sec: max_prob_at_most} without simultaneity. 

\begin{lemma}\label{thm:sim_transformation}
    Consider a non-bifurcatable at-most-$r$-furcating tree $T$, and suppose $T'$ is obtained from $T$ by bifurcatization. Then, allowing simultaneity, $N(T') > N(T)$.
\end{lemma}

We omit the proof of this result, which is almost entirely identical to that of Lemma \ref{thm:transform_more_LH}. The statement of this result and its proof follow those of Lemma 6 with two slight changes. First, in the construction of internal node $k$ in the bifurcatization, we also increase the number of events by 1. Second, in the step that shows that the inequality $N(T') > N(T)$ is strict, we begin with a tie-permitting labeled history of $T$ in which $c_1, c_2, \ldots, c_j$ are all in the same event. With node $k$ and the added event introduced, one tie-permitting labeled history of $T'$ places $c_3$ and $k$ in the same event and $c_1, c_2$ in a separate event, and another places $k$ as its own event and keeps $c_1, c_2, c_3$ together in a separate event. Note that this construction also verifies that $E(T, z) < E(T', z+1)$.

\begin{theorem}
\label{thm:sim_max_bi}
    For $r \geq 2$, and $n \geq 1$, an at-most-$r$-furcating tree shape with the largest number of tie-permitting labeled histories among at-most-$r$-furcating tree shapes with $n$ leaves is a strictly bifurcating tree.
\end{theorem}

The statement of this result and its proof follow those of Theorem \ref{thm:bi_max_prob} with two slight changes. First, for the proof, Lemma \ref{thm:sim_transformation} is used in place of Lemma \ref{thm:transform_more_LH}. Second, the statement of the result recognizes that we have not yet established that a unique at-most-$r$-furcating tree shape has the largest number of tie-permitting labeled histories. In the proof, if there are multiple strictly bifurcating tree shapes with the largest number of tie-permitting labeled histories, then choose one of them arbitrarily for the role of $\hat{T}_n$ in the proof.

We have therefore reduced the problem of finding the $n$-leaf at-most-$r$-furcating unlabeled topology whose labelings have the largest number of tie-permitting labeled histories to the problem of finding the $n$-leaf bifurcating unlabeled topology (or topologies) whose labelings have the largest number of tie-permitting labeled histories. We have not characterized this bifurcating unlabeled topology, but we can state a conjecture. For $2 \leq n \leq 21$, with $z$ ranging in $\lceil \log_2(n) \rceil \leq z \leq n-1$, Tables \ref{tab:conjecture_bi_sim} and \ref{tab:conjecture_bi_sim_2} report the maximal number of tie-permitting labeled histories with $z$ events across all bifurcating unlabeled topologies with $n$ leaves. For $n=1$, one labeled history occurs with $z=0$, and 0 labeled histories occur for all other $z$. 

In all cases of $(n,z)$ in the tables, the unlabeled topology in Theorem \ref{thm: hammersley} produces the maximal number of tie-permitting labeled histories. For some $(n,z)$, we find that multiple unlabeled topologies share this maximal number. However, with fixed $n$, this sharing of the maximum does not occur at all permissible $z$. We are led to the following conjectures.
\begin{conjecture} 
\label{conj:max_bi_sim}
Consider the set of rooted bifurcating unlabeled topologies with $n$ leaves, permitting simultaneous bifurcations. 
\begin{enumerate}[label=(\roman*)]
    \item The unlabeled topology whose labelings have the largest number of tie-permitting labeled histories takes the form in Theorem \ref{thm: hammersley}. This topology is unique in having the maximal value.
    \item Further, this same unlabeled topology has the largest number of tie-permitting labeled histories with exactly $z$ events, $\lceil \log_2 n \rceil \leq z \leq n-1$. This topology is not necessarily unique in having the maximal value.
\end{enumerate}
\end{conjecture}

Note that part (ii) of the conjecture, together with Theorem \ref{thm: hammersley}, implies part (i). To demonstrate part (i), it suffices to show show part (ii) and to exhibit some value of $z$ for which the topology specified by Theorem \ref{thm: hammersley} uniquely achieves the maximum at that $z$. The required value of $z$ is $z=n-1$: with $z=n-1$, no ties occur, and by Theorem \ref{thm: hammersley}, the tree specified by the proposition uniquely achieves the maximal number of tie-permitting labeled histories with $n$ leaves and $z=n-1$ events.

As stated in part (ii), the maximum does not always occur for a unique unlabeled topology. For example, the two tree shapes in Figure \ref{fig:non_unique_bi} produce the maximal number of tie-permitting labeled histories for $(n,z)=(13, 4)$, namely 2. This pattern of non-uniqueness generalizes. For each $k \geq 4$, for $(n,z)=(2^k - 3, k)$, non-uniqueness occurs with two trees of similar structure. One of the trees---call it $T_1$---has subtrees $L_1$, $R_1$ of sizes $(|L_1|,|R_1|)=(2^{k-1} - 1, 2^{k-1} - 2)$. The other, $T_2$, has subtrees $L_2, R_2$ of sizes $(|L_2|,|R_2|)=(2^{k-1}, 2^{k-1}-3)$. In the two trees, subtrees $L_1, R_1, L_2, R_2$ have the topologies in Theorem \ref{thm: hammersley}. Trees $T_1$ and $T_2$ both produce 2 tie-permitting labeled histories; in both trees, one internal node can be labeled 1 or 2, and all other internal nodes have a fixed label because they lie on a length-$k$ path between leaves and the root. 

\begin{figure}
    \centering
    \includegraphics[width=0.6\linewidth]{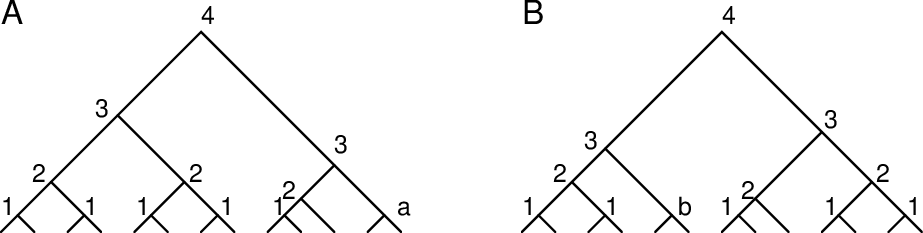}
    \caption{The two rooted bifurcating unlabeled topologies whose labelings produce the maximal number of tie-permitting labeled histories for $(n,z)=(13,4)$. Both trees both produce 2 tie-permitting labeled histories. Internal nodes are labeled by the events to which they are assigned. (A) The topology in Theorem \ref{thm: hammersley}. Two tie-permitting labeled histories are possible, with $a=1$ and $a=2$. (B) An alternative topology. Two tie-permitting labeled histories are possible, with $b=1$ and $b=2$.}
    \label{fig:non_unique_bi}
\end{figure}

\section{Discussion}

In this paper, we have extended results on the enumeration of labeled histories from strictly $r$-furcating trees to at-most-$r$-furcating trees. With non-simultaneous branching, we enumerated the total number of at-most-$r$-furcating trees on $n$ leaves (Proposition \ref{eq: num_at_most_r}) and the number of labeled histories for a specific topology (Proposition \ref{prop:non-sim_specific}). We then identified the at-most-$r$-furcating  tree with the largest number of labeled histories as the maximally probable bifurcating tree (Theorem \ref{thm:bi_max_prob}). 

Next, allowing simultaneous branching, we enumerated the total number of tie-permitting labeled histories for at-most-$r$-furcating trees on $n$ leaves (Proposition \ref{eq:arb_at_most_r_sim}) and the number of tie-permitting labeled histories for a specific topology (Proposition \ref{eq:sim_for_at_most_r}). We then reduced the problem of identifying the at-most-$r$-furcating tree with the largest number of labeled histories to that of finding the bifurcating tree with the largest number of tie-permitting labeled histories (Theorem \ref{thm:sim_max_bi}). 

This work continues recent extensions of enumerative phylogenetic results from standard settings of non-simultaneity and bifurcation to allow multifurcation and simultaneity. Recent interest in multifurcation traces to its potential applicability in pathogen transmission models, rapid speciation, and genealogical models for organisms in which some individuals have very large numbers of offspring~\cite{Eldon20, Maranca2023, Wakeley08}. Such settings can also possess simultaneous coalescence, so that simultaneity is of interest in extended models as well~\cite{Dickey25, King23}. The consideration of at-most-$r$-furcation here extends beyond the strict $r$-furcation in \cite{Dickey25}. 

The mathematics of labeled histories has correspondences with concepts in other settings. King \& Rosenberg \cite{King23} discussed how the labeled histories of a labeled topology correspond to sequences in which the games in a single-elimination tournament can be played. If simultaneity is permitted, then tie-permitting labeled histories specify tournament schedules allowing for multiple games to be played simultaneously in different arenas. More generally, the area of operations research considers precedence constraints for task scheduling~\cite{Pinedo22}, and a labeled history corresponds to a set of precedence constraints; simultaneity in labeled histories is analogous to availability of multiple machines in operations research. Just as labeled topologies vary in their numbers of labeled histories, sets of precedence constraints vary in their numbers of valid schedulings~\cite{Prot18}; problems of finding maximally probable labeled topologies correspond to problems of finding, for a set of tasks, the precedence constraints that give rise to the largest number of valid schedulings.

A number of computations in this study recapitulate and generalize known number sequences. For example, in Table 1, $A_2(n)$ in the $r=2$ column recovers sequence A006472, $A_3(n)$ recovers sequence A358072, and the diagonal $A_n(n)$ recovers sequence A256006. In Table 2, the $r=2$ column, $S_2(n)$, corresponds to OEIS sequence A317059, and the diagonal, $S_n(n)$, corresponds to sequence A005121. Tables 3 and 4 (transposed) correspond to sequence A378855, with column sums following A380767.

An outstanding problem in this study is that of identifying the (possibly non-unique) bifurcating tree shape whose labelings possess the largest number of tie-permitting labeled histories. Based on computations up to $n=21$ (Tables \ref{tab:conjecture_bi_sim} and \ref{tab:conjecture_bi_sim_2}), in Conjecture \ref{conj:max_bi_sim}, we have conjectured that the (unique) bifurcating tree shape that produces the maximal number of tie-permitting labeled histories follows Theorem \ref{thm: hammersley}. A related claim, which we have demonstrated in Theorem \ref{thm:sim_max_bi}, is that the (possibly non-unique) at-most-$r$-furcating tree shape with the most tie-permitting labeled histories is a bifurcating shape. Hence, if the conjecture holds, then the (unique) at-most-$r$-furcating tree shape with the most tie-permitting labeled histories also has the bifurcating shape in Theorem \ref{thm: hammersley}. 

\begin{table}[tb]
    \centering
    \footnotesize
    \begin{tabular}{|c|c|c|c|c|c|c|c|c|c|c|c|c|c|c|}
        \hline
          & \multicolumn{14}{c|}{Number of leaves, $n$}  \\ 
        \cline{2-15}
        $z$   & 2 & 3 & 4 & 5 &  6 &  7 &   8 &     9 &    10 &     11 &      12 &        13 &         14 &          15 \\ \hline
        1     & 1 & 0 & 0 & 0 &  0 &  0 &   0 &     0 &     0 &      0 &       0 &         0 &          0 &           0 \\ \hline
        2     & 0 & 1 & 1 & 0 &  0 &  0 &   0 &     0 &     0 &      0 &       0 &         0 &          0 &           0 \\ \hline
        3     & 0 & 0 & 2 & 2 &  2 &  1 &   1 &     0 &     0 &      0 &       0 &         0 &          0 &           0 \\ \hline
        4     & 0 & 0 & 0 & 3 &  9 & 12 &  22 &    10 &    10 &      5 &       5 &         2 &          2 &           1 \\ \hline
        5     & 0 & 0 & 0 & 0 &  8 & 30 & 102 &   114 &   198 &    204 &     344 &       278 &        434 &         412 \\ \hline
        6     & 0 & 0 & 0 & 0 &  0 & 20 & 160 &   380 & 1,100 &  1,930 &   4,890 &     6,360 &     14,016 &      20,130 \\ \hline
        7     & 0 & 0 & 0 & 0 &  0 &  0 &  80 &   485 & 2,495 &  7,260 &  27,110 &    53,000 &    159,560 &     321,820 \\ \hline
        8     & 0 & 0 & 0 & 0 &  0 &  0 &   0 &   210 & 2,478 & 12,810 &  72,702 &   211,365 &    866,775 &   2,390,150 \\ \hline
        9     & 0 & 0 & 0 & 0 &  0 &  0 &   0 &     0 &   896 & 10,640 & 101,024 &   451,164 &  2,572,052 &   9,685,872 \\ \hline
        10    & 0 & 0 & 0 & 0 &  0 &  0 &   0 &     0 &     0 &  3,360 &  70,080 &   529,116 &  4,408,404 &  23,150,064 \\ \hline 
        11    & 0 & 0 & 0 & 0 &  0 &  0 &   0 &     0 &     0 &      0 &  19,200 &   321,600 &  4,357,632 &  33,549,120 \\ \hline 
        12    & 0 & 0 & 0 & 0 &  0 &  0 &   0 &     0 &     0 &      0 &       0 &    79,200 &  2,307,360 &  28,979,280 \\ \hline
        13    & 0 & 0 & 0 & 0 &  0 &  0 &   0 &     0 &     0 &      0 &       0 &         0 &    506,880 &  13,728,000 \\ \hline
        14    & 0 & 0 & 0 & 0 &  0 &  0 &   0 &     0 &     0 &      0 &       0 &         0 &          0 &   2,745,600 \\ \hline
        Total & 1 & 1 & 3 & 5 & 19 & 63 & 365 & 1,199 & 7,177 & 36,209 & 295,355 & 1,652,085 & 15,193,115 & 114,570,449 \\ \hline
    \end{tabular}
\caption{For $2 \leq n \leq 15$ leaves, the number of tie-permitting labeled histories for the maximally probable bifurcating tree shape, allowing simultaneity. Columns correspond to the number of leaves, $n$, and rows to the number of events, $z$, $\lceil \log_2 n \rceil \leq z \leq n-1$. Each entry is found by computing $E(T, z)$ using Proposition \ref{eq:sim_for_at_most_r}, where $r=2$, for all strictly bifurcating unlabeled tree shapes on $n$ leaves, and taking the maximum. The maximizing tree shape is, in all $(n,z)$ in the table, the shape in Theorem \ref{thm: hammersley}, but that shape is not necessarily the only maximizing shape for an entry $(n,z)$. The maximizing total, summing across rows, is unique. The table of values for $(n,z)$ corresponds to OEIS 378855 transposed; the total corresponds to OEIS A380767.}
\label{tab:conjecture_bi_sim}
\end{table}

\begin{table}[tb]
    \centering
    \footnotesize
    \vspace{1cm}
    \begin{tabular}{|c|c|c|c|c|c|c|}
        \hline
         & \multicolumn{6}{c|}{Number of leaves, $n$}  \\ \cline{2-7}
        $z$   &            16 &            17 &             18 &              19 &                 20 &                 21 \\ \cline{1-7}
        4     &             1 &             0 &              0 &               0 &                  0 &                  0 \\ \cline{1-7}
        5     &           672 &           260 &            260 &             130 &                130 &                 52 \\ \cline{1-7}
        6     &        45,914 &        35,108 &         53,756 &          50,188 &             81,268 &             63,676 \\ \cline{1-7}
        7     &       973,300 &     1,147,560 &      2,409,000 &       3,319,860 &          7,396,980 &          8,658,240 \\ \cline{1-7}
        8     &     9,396,760 &    15,642,395 &     42,972,365 &      80,327,145 &        231,570,595 &        366,969,220 \\ \cline{1-7}
        9     &    49,410,424 &   111,849,304 &    393,883,672 &     960,564,444 &      3,480,089,340 &      7,122,959,508 \\ \cline{1-7}
        10    &   155,188,488 &   471,859,668 &  2,117,397,324 &   6,617,863,308 &     29,725,413,060 &     76,673,425,752 \\ \cline{1-7}
        11    &   304,369,008 & 1,250,312,856 &  7,186,950,312 &  28,651,896,456 &    158,936,626,776 &    510,467,689,056 \\ \cline{1-7}
        12    &   376,231,680 & 2,140,177,050 & 16,024,041,990 &  81,957,989,850 &    563,604,210,510 &  2,246,305,636,905 \\ \cline{1-7}
        13    &   284,951,040 & 2,365,158,180 & 23,815,148,060 & 158,943,918,980 &  1,370,645,607,980 &  6,805,224,583,410 \\ \cline{1-7}
        14    &   120,806,400 & 1,630,311,540 & 23,382,250,892 & 210,260,550,500 &  2,321,850,953,708 & 14,524,535,914,746 \\ \cline{1-7}
        15    &    21,964,800 &   637,665,600 & 14,570,322,624 & 186,971,378,880 &  2,737,984,132,416 & 22,032,147,914,088 \\ \cline{1-7}
        16    &             0 &   108,108,000 &  5,222,016,800 & 106,955,008,160 &  2,204,742,219,680 & 23,627,935,026,696 \\ \cline{1-7}
        17    &             0 &             0 &    820,019,200 &  35,568,332,800 &  1,156,376,166,400 & 17,515,815,233,280 \\ \cline{1-7}
        18    &             0 &             0 &              0 &   5,227,622,400 &    356,157,235,200 &  8,541,180,560,640 \\ \cline{1-7}
        19    &             0 &             0 &              0 &               0 &     48,881,664,000 &  2,465,468,928,000 \\ \cline{1-7}
        20    &             0 &             0 &              0 &               0 &                  0 &    319,258,368,000 \\ \cline{1-7}
        Total & 1,323,338,487 & 8,732,267,521 & 93,577,466,255 & 822,198,823,101 & 10,952,623,368,043 & 98,672,511,931,269 \\ 
        \hline
    \end{tabular}
\caption{For $16 \leq n \leq 21$ leaves, the number of tie-permitting labeled histories for the maximally probable bifurcating tree shape, allowing simultaneity. The table continues Table \ref{tab:conjecture_bi_sim}, whose design it follows.}
\label{tab:conjecture_bi_sim_2}
\end{table}

\bigskip
\noindent {\bf Acknowledgments.} Grant support was provided by National Science Foundation grant DMS-2450005.

\bibliographystyle{plain}
\bibliography{multifurcating2}
\clearpage

\end{document}